\documentclass[11pt,a4paper,reqno,twoside]{amsart}

\usepackage[dvipsnames]{xcolor}
\usepackage{float}
\usepackage{latexsym}
\usepackage{amsthm}
\usepackage{empheq}
\usepackage{bm}
\usepackage{booktabs}
\usepackage{subcaption}
\usepackage{tikz,pgfplots,cite}
\usepackage{balance}

\usepackage{cite}
\usepackage{amsmath,amssymb,amsfonts}
\usepackage{algorithm}
\usepackage{algorithmic}
\usepackage{graphicx}

\usepackage{latexsym}
\usepackage{amsthm}
\usepackage{thm-restate}
\usepackage{empheq}
\usepackage{bm}
\usepackage{booktabs}
\usepackage{subcaption}
\usepackage{enumitem}
\usepackage{array}
\usepackage[margin=2.8cm]{geometry}

\definecolor{redish}{rgb}{0.9, 0.17, 0.31}
\definecolor{fuchs}{rgb}{0.57, 0.36, 0.51}
\usepackage[colorlinks=true,linkcolor=fuchs]{hyperref}

\theoremstyle{definition}
\newtheorem{theorem}{Theorem}

\newtheorem{proposition}{Proposition}
\newtheorem{definition}{Definition}
\newtheorem{example}{Example}

\newtheorem{notation}{Notation}

\newtheorem{remark}{Remark}
\newtheorem{question}{Question}
\newtheorem{lemma}{Lemma}

\usepackage{comment}

\newcommand{\F}{\mathbb{F}}

\newcommand{\mC}{\mathcal{C}}

\newcommand{\mA}{\mathcal{A}}

\newcommand{\supp}{\operatorname{supp}}

\newtheorem{claim}{Claim}

\def\BibTeX{{\rm B\kern-.05em{\sc i\kern-.025em b}\kern-.08em
    T\kern-.1667em\lower.7ex\hbox{E}\kern-.125emX}}

\newcommand{\bX}{{\bf X}}

\newcommand{\balpha}{{\boldsymbol\alpha}}
\newcommand{\bc}{{\bf c}}
\newcommand{\bs}{{\bf s}}
\newcommand{\bzero}{{\bf 0}}

\newcommand{\Fq}{\mathbb{F}_q}
\newcommand\ed[2]{\textup{ED}(#1,#2)}

\newcommand\RSS{\text{RS}_{n,k}(\balpha)}
\newcommand{\RSfull}{\text{RS}_{q,k}(\alpha_1, \ldots, \alpha_q)}

\title[RS Codes Against Insertions and Deletions: Full-Length and Rate-$1/2$ Codes]{Reed-Solomon Codes Against Insertions and Deletions: Full-Length and Rate-$1/2$ Codes}
\usepackage[foot]{amsaddr}

\author{Peter Beelen$^1$}
\address{$^1$Technical University of Denmark, Denmark.}
\thanks{R. C. and E. Y. are supported in part by the Israel Science Foundation (ISF) grant 2462/24 and are supported by the European Union (DiDAX, 101115134).
Views and opinions expressed are those of the author(s) only and do not necessarily reflect those of the European Union or the European Research Council Executive Agency. Neither the European Union nor the granting authority can be held responsible for them. P. B., A. G. and M. M. are supported by the Villum Fonden through grant VIL”52303”.}
\email{\{pabe,anigr,marimo\}@dtu.dk, roni.con93@gmail.com, yaakobi@cs.technion.ac.il}

\author{Roni Con$^2$}
\address{$^2$Technion -- Israel Institute of Technology, Israel.}

\author{Anina Gruica$^1$}

\author{Maria Montanucci$^1$}

\author{Eitan Yaakobi$^2$}

\usepackage{setspace}
\begin{document}
\maketitle

\begin{abstract}
    The performance of Reed--Solomon codes (RS codes, for short) in the presence of insertion and deletion errors has attracted growing attention in recent literature. In this work, we further study this intriguing mathematical problem, focusing on two regimes. First, we study the question of how well \emph{full-length} RS codes perform against insertions and deletions. 
    For 2-dimensional RS codes, we provide a complete characterization of codes that cannot correct even a single insertion or deletion. Furthermore, we prove that for sufficiently large field size~$q$, nearly all full-length $2$-dimensional RS codes can correct up to $(1 - \delta)q$ insertion and deletion errors for any $0 < \delta < 1$. Extending beyond the 2-dimensional case, we show that for any $k \ge 2$, there exists a full-length $k$-dimensional RS code capable of correcting $q / (10k)$ insertion and deletion errors, provided $q$ is large enough. Second, we focus on rate $1/2$ RS codes that can correct a single insertion or deletion error. We present a polynomial-time algorithm that constructs such codes over fields of size $q = \Theta(k^4)$. This result matches the existential bound given in \cite{con2023reed}.
\end{abstract}

\bigskip
\bigskip
\renewcommand{\baselinestretch}{1.05}\normalsize

\tableofcontents

\newpage
\section{Introduction}
This work considers the model of insertions and deletions (insdel errors, for short), the most common model of synchronization errors. An insertion error occurs when a new symbol is inserted between two adjacent symbols of the transmitted word, and a deletion is when a symbol is removed from the transmitted word. 
Note that these types of errors, unlike substitutions or erasures, can change the length of the transmitted message. 
Insdel errors cause loss of synchronization between the sender and the receiver, which makes the task of designing codes (over small alphabets) for this model a very tantalizing question.

This natural theoretical model, together with possible application in various fields, including in DNA-based storage systems, has led many researchers to construct and study codes against insdel errors (most of the results can be found in the following excellent surveys \cite{cheraghchi2020overview,haeupler2021synchronization}). 


In this work, we focus on one of the most well-known family of linear codes, called Reed--Solomon codes (RS codes), which are defined as follows.
\begin{definition}[Reed--Solomon code]
    Let $\alpha_1, \alpha_2, \ldots, \alpha_n$ be distinct elements in the finite field $\mathbb{F}_q$ with $q$ elements. For $k<n$, the $[n,k]_q$ \emph{Reed--Solomon} (\textit{RS}) code of dimension~$k$ and block length~$n$ associated with the evaluation vector $\balpha = (\alpha_1, \ldots, \alpha_n)\in\F_q^n$ is defined as the $\F_q$-linear space
    \[
    \RSS := \left \lbrace \left( f(\alpha_1), \ldots, f(\alpha_n) \right) : f\in \mathbb{F}_q[x]_{<k} \right \rbrace \subseteq \F_q^n.
    \]	
\end{definition}
RS codes play an important role in ensuring data integrity across various media. Several applications include QR codes, distributed data storage, and data transmission over noisy channels. Furthermore, RS codes have found many theoretical applications in cryptography and theoretical computer science. 
The appeal of RS codes is due to their simple algebraic structure, which provides efficient encoding and decoding algorithms, and also shows that they have optimal rate-error-correction trade-off in the \emph{Hamming} metric. Thus, it is natural to ask how well RS codes perform against insdel errors. This has been studied in several papers \cite{safavi2002traitor,wang2004deletion,tonien2007construction,duc2019explicit,liu20212,con2023reed,con2024optimal,liu2024optimal,con2024random}. However, many unsolved questions still remain and there is much to discover. 

In this paper, we study two regimes of RS codes.
First, we focus on \emph{full-length} RS codes, that is, $q = n$ and the evaluation vector $\balpha = (\alpha_1, \ldots, \alpha_n)\in\F_q^q$ corresponds to some permutation of the elements of $\F_q$. Contrary to RS codes against classical errors (i.e., substitution errors), permuting coordinates of an RS code can significantly reduce (or increase) the number of insertion or deletion errors it can correct. 

We fully characterize ``bad'' permutations on the elements of $\Fq$ for which the corresponding $2$-dimensional RS codes cannot correct any deletion or insertion.
Furthermore, we prove that for sufficiently large field size~$q$, \emph{almost all} full-length 2-dimensional RS codes can correct up to $(1 - \delta)q$ insertion and deletion errors for any constant $0 < \delta < 1$. In the more general $k$-dimensional setting, we show that for $q = 2^{O(k)}$, there exists an $[q,k]_q$ RS code that can correct $q/(10k)$ insdel errors. 

Second, we focus on rate $1/2$ RS codes. We give an algorithm that runs in polynomial time in~$k$, that constructs a $[2k,k]_q$ RS code that corrects a single insdel error, where $q = O(k^4)$. We note that the explicit constructions in \cite{con2023reed,liu2024optimal} require a field of size $2^{k^k}$ and that our result matches the existential result from~\cite{con2023reed}.

\subsection{Previous Work}

\paragraph{\textit{Non-linear insdel codes}}
The study of codes that can correct adversarial insertions and deletions (synchronization errors) started with the seminal works of Levenshtein \cite{levenshtein1966binary}, who showed that the codes of Varshamov and Tenengolts \cite{varshamov1965codes} (correcting asymmetric errors) are optimal binary codes that can correct a single insdel error.
The quest for constructing (close to) optimal codes that can correct a \emph{constant} number of insertions or deletions (even just two) spans many works in recent years with some astonishing results \cite{gabrys2018codes,sima2019two,brakensiek2017efficient,sima2020optimal,sima-q2020optimal,guruswami2021explicit,liu2024explicit}. Despite all of this progress, the question of determining the optimal redundancy-error trade-off for codes correcting a constant number of deletions is still open. 
When it comes to correcting a \emph{fraction} of insdel errors, Haeupler and Shahrasbi \cite{haeupler2021synchronization-org} presented an efficient code with rate $1 - \delta - \varepsilon$ over an alphabet of size $O_{\varepsilon}(1)$ that can correct a $\delta$ fraction of insdel errors. 
These codes are optimal in the sense that they can get as close as we want to the \emph{Singleton bound}, which is $1 - \delta$. 
For \emph{binary} codes, the gap between the upper \cite{yasunaga2024improved} and lower bound \cite{levenshtein2002bounds} on the rate-error trade-off is huge. 


\paragraph{\textit{Linear insdel codes and RS codes}}
Reading the above paragraph, one might ask: ``How come no code among the referenced ones is a linear code?''.
The reason appears in \cite{abdel2007linear} where it was shown that \emph{any} linear code correcting a single insdel error must have rate at most $1/2$. This shows that linear codes are provably worse than non-linear codes as non-linear codes correcting a single insdel error can have rate $1- o(1)$. 
Then, in \cite{cheng2023efficient}, the authors generalized this result and proved the \emph{half-Singleton bound} which states that any $[n,k]_q$ linear code can correct at most $n - 2k + 1$ insdel errors. 
More upper bounds on special families of linear codes correcting insdel errors are given in \cite{chen2022coordinate,ji2023strict,xie2024new}. In particular, in \cite{ji2023strict} the authors show that if the all-$1$ codeword is contained in an $[n,k]$ linear code, then it can correct at most $n-2k$ insdel errors. 

In this work, we focus on RS codes which, by definition, require the alphabet size $q$ to be at least $n$. 
As far as we know, the performance of RS codes against insdel errors was first considered in \cite{safavi2002traitor} in the context of the traitor tracing problem. In \cite[Theorem 3.2]{wang2004deletion}, the authors constructed a $[5,2]_q$ \emph{generalized} RS code that can correct a single deletion when $q > 8$. They also showed how to extend this construction for any~$k$ (by induction), but only provided that there exists a choice of the evaluation points $\alpha_i$ and multipliers $v_i$ (which are the non-zero field elements used to scale the coordinates of the codewords) satisfying some specific properties \cite[Theorem 3.3]{wang2004deletion}. The resulting rate $1/2$ codes are never usual RS codes. This can be seen in \cite[Theorem 3.2]{wang2004deletion} where in fact the choice $v_i=1$ for all $i$ is not allowed. In \cite{tonien2007construction}, the authors constructed~$[n,k]_q$ RS codes (where $n$ can be as large as $q$) that can correct $\log_{k}(q) - 1$ insdel errors. However, their construction relies on the existence of a polynomial of degree~$k$ that satisfies special properties. The authors do not provide a method for constructing such a polynomial, nor do they show its existence (for general $k$ and $q$).

Then, in \cite{con2023reed}, the authors showed the existence of $[n,k]_q$ RS codes that can correct $n-2k+1$ insdel errors where $$q = O\left(\binom{n}{2k-1}^2 k^2\right).$$ These codes attain the half-Singleton bound with equality. They also provided an explicit construction of such codes where $q \approx n^{k^k}$. In \cite{con2024random}, it was shown that there are $[n,k]_q$ RS codes that can correct $(1-\varepsilon)n - 2k + 1$ insdel errors where $q = O(n + 2^{\textup{poly}(1/\varepsilon)}k)$.
In \cite{duc2019explicit,liu20212,con2023reed,con2024optimal}, the particular case of $k=2$ was studied and several explicit constructions were given. It was shown that the minimum field size of an $[n,2]_q$ correcting $n-3$ deletions is $\Omega(n^3)$ and it is accompanied by an explicit construction with $q=O(n^3)$ \cite{con2024optimal}.

\subsection{Our Contribution}
In this work, we consider two very basic questions regarding the performance of RS codes against insdel errors.
\begin{question}
    Can a full-length RS code correct insdel errors, and if it can, how many?
\end{question}
Consider the scenario where our RS code is defined over $\F_p$, for a prime $p$, and $\balpha = (0,1,2,\ldots,p-1)$. Then this $\RSS$ code over $\F_p$ cannot correct even a single deletion. 
In fact, if we remove the first symbol from the codeword that corresponds to $f(x) = x$ and the last symbol from the codeword that corresponds to $g(x) = x+1$, then we get the same vector of length $p-1$. 
However, what happens if we consider ``less natural'' orderings on the points? Can we choose a different permutation of $\F_p$ for which the corresponding 2-dimensional RS code can correct a single deletion -- or perhaps even more? 

We answer this question in the affirmative. 
First, we give a complete characterization of all orderings $\balpha= (\alpha_1,\ldots,\alpha_q)$ of $\Fq$ that give rise to 2-dimensional RS codes that cannot correct even a single deletion. The number of such bad orderings is tiny compared to the total number of possible orderings, $q!$. We show that, even for $q=7$, more than $95\%$ of all orderings yield a 2-dimensional RS code correcting $1$ insdel. 

Second, we show that, in fact, most $2$-dimensional RS codes can correct any linear fraction of insdel errors, provided that $q$ is large enough. Specifically, we show the following:

\begin{restatable}{theorem}{twodim} \label{thm:2-dim}
    Let $\varepsilon,\delta > 0$. Then, for every prime power $q>q_0(\delta,\varepsilon)$, 
    at least $(1-\varepsilon)$ fraction of all $\textup{RS}_{q,2}$ codes can correct $(1-\delta) q$ insdel errors. 
\end{restatable}

In particular, our result implies that a uniformly random full-length 2-dimensional RS code will, with high probability, correct any number of insdel errors which is linear in~$q$.
Beyond the 2-dimensional setting, we show -- using a probabilistic argument -- that for large enough $q$, there exists an ordering of $\F_q$ such that the respective RS code can correct $q/(16k)$ insdel errors. Formally, we prove the following theorem.

\begin{restatable}{theorem}{existperm}\label{thm:main}
    Let $k$ be an integer and $q$ be a prime power such that $q \geq e^{6k}\cdot (10 e k^3)$. 
    There exists an evaluation vector $\balpha = (\alpha_1, \ldots, \alpha_{q})$ such that the code $\textup{RS}_{q,k}(\balpha)$ can correct any $q/(10k)$ insdel errors. 
\end{restatable}

The second question considers rate-$1/2$ Reed-Solomon codes that can correct a single deletion. Recall that any linear code correcting a single insdel error must have rate at most $1/2$ and for $k=1$ it is achieved by the trivial repetition code. For larger dimensions, it is not obvious which codes of rate $1/2$ can correct a single deletion.  
\begin{question}
    Construct RS codes of rate exactly $1/2$ that can correct a single insdel error. 
\end{question}

We start by presenting a $[4,2]_7$ RS code that can correct a single insdel error and show that $q=7$ is the minimal field size for such a code. Then, this code will be used for an inductive process in which we construct an $[2k,k]_q$ RS code correcting a single insdel error. Formally:
 
\begin{restatable}{theorem}{AlgHalfCode} \label{thm:alg}
    Let $q = O(k^4)$ be a prime power. There exists a polynomial time algorithm that outputs $\balpha = (\alpha_1, \ldots, \alpha_{2k}) \in \F_q^{2k}$ for which the respective $\RSS$ code can correct a single insdel error.
\end{restatable}

We emphasize that this result matches the existential bound on the field size given in \cite{con2023reed}. 

\section{Preliminaries}
Throughout, $\F_q$ will denote a finite field of order $q$ and $\Fq[X_1, \ldots, X_n]$ will denote the ring of polynomials in $X_1,\dots,X_n$ over $\Fq$.
We recall the notions of a subsequence and a longest common subsequence. 
\begin{definition}
        A \emph{subsequence} of a string $\bs$ is a string obtained by removing some (possibly none) of the symbols in $\bs$. 
\end{definition}

\begin{definition}
        Let $\bs,\bs'$ be strings (of possibly different lengths) over an alphabet $\Sigma$. A \emph{longest common subsequence} between $\bs$ and $\bs'$, is a subsequence of both $\bs$ and $\bs'$, of maximal length. We denote by $ \textup{LCS}(\bs,\bs')$ the length of a longest common subsequence of $\bs$ and $\bs'$.
	
        The \emph{edit distance} or \emph{insdel distance} between $\bs$ and $\bs'$, denoted by $\ed{\bs}{\bs'}$, is the minimum number of insertions and deletions needed to transform $\bs$ into $\bs'$. 
\end{definition}

\begin{example}
Consider the following vectors over $\F_5$:
\begin{align*}
\bs = (2, 4, 1, 3, 0, 2) \quad \text{and} \quad \bs' = (4, 3, 2, 1, 0).
\end{align*}
A common subsequence of $\bs$ and $\bs'$ is $(4, 3, 0)$:
\begin{align*}
(2, \underline{4}, 1, \underline{3}, \underline{0}, 2) \quad \text{and} \quad (\underline{4}, \underline{3}, 2, 1, \underline{0}).
\end{align*}
Another is $(4, 1, 0)$, and both have length 3. One can check that no longer common subsequence exists. Thus, the length of a longest common subsequence is $\textup{LCS}(\bs, \bs') = 3.$
\end{example}

For a code $\mC$ we use the following notations:
\begin{align*}
    \textup{LCS}(\mC) &:= \max\{\textup{LCS}(\bc,\bc') : \bc,\bc' \in \mC, \bc \ne \bc'\}, \\
    \textup{ED}(\mC) &:= \min\{\textup{ED}(\bc,\bc') : \bc,\bc' \in \mC, \bc \ne \bc'\}.
\end{align*} 
We have $\textup{ED}(\mC)=2n-2\textup{LCS}(\mC)$ if $\mC$ has length $n$.
It is well known that the insdel correction capability of a code is determined by the LCS of its codewords. Specifically,
\begin{lemma} \label{lem:code-lcs}
    A code $\mC$ can correct $\delta$ insdel errors if and only $\textup{LCS}(\bc,\bc')\leq n - \delta - 1$ for any distinct $\bc,\bc'\in \mC$, i.e., $\textup{LCS}(\mC) \leq n - \delta - 1$.
\end{lemma}

We observe that the edit distance $\ed{\cdot}{\cdot}$ satisfies the following for any $\bc,\bc' \in \mathbb{F}_q^n$:
\begin{align} \label{eq:isometrie1}
\ed{\bc}{\bc'} &= \ed{\lambda \bc}{\lambda \bc'}
 \quad \text{for any $\lambda \in \mathbb{F}_q^*$,} \\
\ed{\bc}{\bc'} &= \ed{\bc+\mathbf{1}}{\bc'+\mathbf{1}}, \label{eq:isometrie2}
\end{align}
where $\mathbf{1}=(1,1,\dots,1) \in \mathbb{F}_q^n$.
We denote the Hamming distance between $\bc$ and $\bc'$ by $d_H(\bc,\bc')$ and the Hamming weight of $\bc$ by $w_H(\bc)$.

In this paper, we are interested in RS codes and their insertion deletion correction capabilities. We start by citing the (non-asymptotic version of) rate-error-correction trade-off for \emph{linear} codes correcting insdel errors.
    \begin{theorem}[\textnormal{Half-Singleton bound; see~\cite[Corollary 5.2]{cheng2020efficient}}]  \label{thm:rsopt}
        An $[n,k]_q$ linear code $\mC$ can correct at most $n - 2k + 1$ insdel errors.
        Equivalently, $\textup{LCS}(\mC) \geq 2k - 2$.
    \end{theorem}
In this paper, we call codes attaining the bound of Theorem~\ref{thm:rsopt} \emph{optimal} codes.

\begin{notation}
We say that a vector of indices $I=(I_1,\dots,I_{\ell})\in [n]^{\ell}$ is an \emph{increasing sequence} if its coordinates are monotonically strictly increasing, i.e., for any  $1\leq i<j\leq \ell$, we have $I_i<I_j$. For an increasing vector $I \in [n]^{\ell}$ and an evaluation vector $\balpha=(\alpha_1,\dots,\alpha_n) \in \F_q^n$, we denote by $\balpha_I$ the subsequence of $\balpha$ indexed by $I$, that is, $\balpha_I := (\alpha_{I_1},\dots,\alpha_{I_{\ell}})$. Moreover, for $f \in \F_q[x]$ we let $f(\balpha_I) := (f(\alpha_{I_1}),\dots,f(\alpha_{I_{\ell}}))$.
\end{notation}

For two increasing sequences $I,J\in [n]^{\ell}$, we define the following matrix of order $\ell\times (2k-1)$ in the formal variables $\bX=(X_1,\ldots,X_n)$, which we denote by $V_{k,\ell, I,J}(\bX)$:
    \begin{align} 
    \label{eq:mat-lcs-eq}
V_{k,\ell, I,J}(\bX) = \begin{pmatrix} 
	1 & X_{I_1} & \ldots & X_{I_1}^{k-1}  & X_{J_1} &\ldots & X_{J_1}^{k-1} \\ 
	1 & X_{I_2} & \ldots & X_{I_2}^{k-1}  & X_{J_2} &\ldots & X_{J_2}^{k-1} \\
	\vdots &\vdots & \ddots &\vdots &\vdots &\ddots &\vdots \\
	1 & X_{I_{\ell}} & \ldots & X_{I_{\ell}}^{k-1}  & X_{J_{\ell}} &\ldots & X_{J_{\ell}}^{k-1}\\
	\end{pmatrix}.
	\end{align}

In~\cite{con2024random}, the following algebraic condition was proved.

\begin{lemma}[see~\textnormal{\cite[Lemma 12]{con2024random}}] \label{lem:alg-cond}
Let $n,k$, and $\ell$ be integers such that $2k-1\leq \ell \leq n$. Consider the $\text{RS}_{n,k}(\balpha)$ code associated with the evaluation points $\balpha=(\alpha_1,\ldots,\alpha_n)$. If the code cannot correct $n-\ell$ insdel errors, then there exist two increasing sequences $I,J\in [n]^{\ell}$ where $d_H(I,J)\geq \ell - k + 1$ such that $\textup{rank} (V_{k,\ell,I,J}(\balpha)) < 2k-1$.
\end{lemma}

We will also need the following notation and lemma. The lemma we state here in a more general setting than what was done in~\cite{Xing}. It can also be seen as a special case of Lemma~\ref{lem:iff}.

\begin{notation} \label{not:ag}
$AGL(\mathbb{F}_q)$ denotes the group of affine maps $f_{a,b}$ with $f_{a,b}(x)=ax+b$, $a \in \mathbb{F}_q^*$ and $b \in \mathbb{F}_q$. In other words, given $x \in \mathbb{F}_q$, $f_{a,b}(x)=ax+b \in \mathbb{F}_q$.
\end{notation}

\begin{lemma} (see also \textnormal{\cite[Lemma 4.10]{Xing}}) \label{lemxing}
$\text{RS}_{n,2}(\balpha)$ has insdel distance $2n - 2\ell$ if and only if $f_{a,b}(\balpha_{I}) \ne {\balpha}_J$ for any $f_{a,b} \in AGL(\mathbb{F}_q)$ and any two increasing vectors $I,J \in [n]^{\ell}$ with $d_H(I,J) \geq \ell-1$.
\end{lemma}

\section{Full-Length RS Codes} 

In this section, we consider full-length RS codes ($q=n$), i.e., the evaluation vector $\balpha=(\alpha_1,\dots,\alpha_q)$ is such that $\{\alpha_1,\dots,\alpha_{q}\} = \F_q$, and can be seen as an ordering or a permutation of $\F_q$. We show how the order in which the elements of $\F_q$ appear in $\balpha$ heavily influences the insdel distance of the code.

\subsection{Characterizing RS Codes which cannot Correct a Single Insdel error}
\begin{definition} \label{def:equi}
    We say that the evaluation vectors $\balpha =(\alpha_1,\dots,\alpha_q) \in \F_q^q$ and $\widetilde{\balpha} =(\widetilde{\alpha}_1,\dots,\widetilde{\alpha}_q)\in \F_q^q$ are \textit{equivalent} if $(\lambda \alpha_1 + \mu, \dots, \lambda \alpha_q + \mu) = \widetilde{\balpha}$ for some $\lambda \in \F_q^*$ and $\mu \in \F_q$.
\end{definition}

It is easy to see that for equivalent $\balpha, \widetilde{\balpha} \in \F_q^q$ we have $\text{RS}_{q,2}(\balpha) = \text{RS}_{q,2}(\widetilde{\balpha})$. Moreover, in every equivalence class, there are $q(q-1)$ vectors (we have $q-1$ choices for $\lambda$ and $q$ choices for $\mu$). Here we used that $\balpha$ is not a constant vector, since it is an evaluation vector.

\begin{lemma} \label{lem:eqiii}
    For a primitive element $\theta$ of $\F_q$, consider the following vectors:
    \begin{enumerate}
        \item $(0,1,\theta,\dots,\theta^{q-2})$;
        \item $(\theta^{q-2}, \dots, \theta, 1, 0)$;
        \item $(0,1,2,\dots, q-1)$ if $q$ is prime.
    \end{enumerate}
    Then, the $\text{RS}_{q,2}(\balpha)$ code cannot correct a single insdel error if and only if~$\balpha$ is equivalent to one of the vectors in $(1)-(3)$ for any $\theta$.
\end{lemma}
\begin{proof}
Suppose that $\balpha$ is such that $\text{RS}_{q,2}(\balpha)$ has insdel distance~2. Then there exist increasing sequences $I,J \in [q]^{q-1}$ with $d_H(I,J)\ge (q-1)-1 = q-2$, and with the property that $f_{a,b}(\balpha_{I}) = {\balpha}_J$ for some $f_{a,b} \in AGL(\F_q)$; see Notation~\ref{not:ag} and Lemma~\ref{lemxing}. 

We start with the case $d_H(I,J)=q-1$. The only possibilities for $I,J \in [q]^{q-1}$ are then $I=(2,\dots,q)$ and $J=(1,\dots,q-1)$, and vice-versa. W.l.o.g. let $I=(2,\dots,q)$ and $J=(1,\dots,q-1)$. By subtracting an appropriate constant vector and multiplying by a scalar, and because equivalent evaluation vectors give rise to the same code, we can further assume that $\balpha=(0,1,\alpha_3,\dots,\alpha_{q})$. Since $f_{a,b}(\balpha_{I}) = {\balpha}_J$ we obtain the following system of equations:
\begin{align*}
    \begin{cases}
        a\cdot 0 + b &= 1 \\
        a\cdot 1 + b &= \alpha_3 \\
        a\cdot \alpha_3 + b &= \alpha_4 \\
        &\vdots \\
        a\cdot \alpha_{q-1} + b &= \alpha_q \\
        a\cdot \alpha_{q} + b &= 0
    \end{cases}
\end{align*}
where the last equation comes from the fact that $f_{a,b}$ acts bijectively on $\F_q$. From the above set of equations we obtain
\begin{align*}
    b=1, \quad \alpha_3=a+1,\quad \dots, \quad \alpha_q=a^{q-2}+\dots+a+1. 
\end{align*}
Since additionally we know that $a \cdot \alpha_q+1=0$, we have that $a^{q-1}+\dots+a+1=0$. If $a\ne 1$, then $\frac{a^q-1}{a-1}=a^{q-1}+\dots+a+1=0$. Therefore $a^q=1$, which implies $a=1$. We obtain a contradiction. Thus, the only solution is $a=1$, and from the above set of equations we obtain that $\balpha=(0,1,2,\dots,q-1)$. Since $\balpha$ is an evaluation vector, this case can only occur if $q$ is a prime.

Now suppose that $d_H(I,J)=q-2$. This implies that either $I$ and $J$ are equal to $(2,\dots, q)$ and $(1,\dots,q-2,q)$, or to $(1,3,\dots, q)$ and $(1,\dots,q-1)$. Assume that $I=(2,\dots, q)$ and $J=(1,\dots,q-2,q)$. As before, through scaling, we can assume that $\balpha=(\alpha_1,\dots,\alpha_{q-2},1,0)$. Since $f_{a,b}(\balpha_{I}) = {\balpha}_J$ we obtain:
\begin{align*}
    \begin{cases}
        a\cdot 0 + b &= 0 \\
        a\cdot 1 + b &= \alpha_{q-2} \\
        a\cdot \alpha_{q-2} + b &= \alpha_{q-3} \\
        &\vdots \\
        a\cdot \alpha_{2} + b &= \alpha_1 \\
        a\cdot \alpha_{1} + b &= 1
    \end{cases}
\end{align*}
where again the last equation comes from the fact that $f_{a,b}$ is a bijection on $\F_q$. We have
\begin{align*}
    b=0,\quad \alpha_{q-2} = a, \quad , \alpha_{q-3} = a^2, \quad \dots, \quad \alpha_1 = a^{q-2}.
\end{align*}
Since $ a\cdot \alpha_{1} + b = 1$ we also have $a^{q-1} = 1$. Because of the assumption that $\balpha$ is an evaluation vector for an RS code over $\F_q$ of length $q$, we need to have $\{\alpha_1,\dots,\alpha_q\} = \{0,1,a,\dots,a^{q-2}\}=\F_q$ and thus $a$ is a primitive element of $\F_q$. Therefore $\balpha$ is equivalent to $(a^{q-2},\dots, a^2,a,1,0)$.

The case where $I=(1,3,\dots, q)$ and $J=(1,\dots,q-1)$ can be proven analogously, where one arrives at the conclusion that $\balpha$ is equivalent to $(0,1,a,a^2,\dots,a^{q-2})$ for some primitive element~$a$ of $\F_q$.

Now suppose that $\balpha$ is equivalent to one of the vectors in $(1)-(3)$. It is enough to show that for any of the vectors listed in this lemma, the 2-dimensional RS code with that vector as the evaluation vector $\balpha$, contains a codeword $\bc$, different from the evaluation vector, which shares a common subsequence of length $q-1$ with the evaluation vector. Let $\theta \in \F_q$ be a primitive element.
    \begin{enumerate}
        \item If $\balpha=(0,1,\theta,\dots,\theta^{q-2})$, then let $\bc = \theta \balpha = (0,\theta,\dots,\theta^{q-2},1)$. Therefore we have $\text{LCS}(\balpha,\bc ) = q-1$.
        \item If $\balpha=(\theta^{q-2}, \dots, \theta, 1, 0)$, then let $\bc  = \theta \balpha = (1, \theta^{q-2},\dots, \theta,0)$. Therefore we have $\text{LCS}(\balpha,\bc ) = q-1$.
        \item If $\balpha=(0,1,\dots, q-1)$, then let $\bc  = \balpha + \textbf{1} = (1, 2, \dots, q-1, 0)$. Therefore we have $\text{LCS}(\balpha,\bc ) = q-1$. \qedhere
    \end{enumerate}
\end{proof} 

In the next proposition, which is a consequence of Lemma~\ref{lem:eqiii}, $\phi$ denotes Euler's totient function, i.e., $\phi(n)$ is the number of positive integers up to $n$ that are relatively prime to $n$.

\begin{proposition} \label{prop:phiphi}
    There are at least 
    \begin{align*}
         \begin{cases}
           (q-2)! -2\phi(q-1) - 1 &\textnormal{ if $q$ is prime,} \\
           (q-2)! -2\phi(q-1)  &\textnormal{ if $q$ is not prime,}
        \end{cases}
    \end{align*}
 equivalence classes in $\F_q^q$ such that for all $\balpha$ in any of these classes the $\text{RS}_{q,2}(\balpha)$ code can correct at least a single insdel error. 
\end{proposition}
\begin{proof}
From Lemma~\ref{lem:eqiii} we know the vectors that give rise to codes with insdel distance 2. Therefore, for every primitive element of $\F_q$ we have 2 equivalence classes, and if $q$ is prime we also have the equivalence class represented by $(0,1,2,\dots,q-1)$. In total there are $(q-2)!$ equivalence classes, which gives the statement of the lemma.
\end{proof}

Note that the lower bound in Proposition~\ref{prop:phiphi} does not account for potential equivalences among the vectors (1) -- (3) from Lemma~\ref{lem:eqiii}. As a result, we may be overcounting the number of \emph{bad} evaluation vectors, leading to a lower bound which is looser than what could possibly be obtained through a more refined analysis.

Some explicit evaluations of the bound in Proposition~\ref{prop:phiphi} are given in Table~\ref{tabli}. We scale the lower bound relative to the total number of (inequivalent) orderings, which is $(q-2)!$.

\begin{table}[h!]
\centering
\caption{Proportion of 2-dimensional full-length RS codes correcting at least a single insdel error within the set of all codes for various values of prime powers $q$. Note that when the lower bound is positive, this shows the existence of a  RS single insdel correcting code.}
\begin{tabular}{|c|c|c|c|c|c|c|c|} 
 \hline
$q$ & 4 & 5 & 7 & 8 & 9 & 11 & 13 \\ \hline
proportion & 0.000 & 0.333 & 0.967 & 0.983 & 0.998 & 0.999 & 0.999
  \\ \hline
    \noalign{\global\arrayrulewidth0.4pt}
\end{tabular}
\label{tabli}
\centering
\end{table}

\subsection{Most 2-dimensional RS Codes can Correct any Fraction of Insdel Errors}

In this subsection, we prove Theorem~\ref{thm:2-dim}.
First, we give a lower bound on the number of $2$-dimensional full-length Reed-Solomon codes over $\F_q$ that are able to correct $q-\ell$ insdel errors. We introduce the following notation and then state the proposition.

\begin{notation}
    For positive integers $n$ and $\ell$ and a sequence $I=(I_1,\dots,I_{\ell}) \subseteq [n]^{\ell}$, we will denote by $\supp(I)$ the set made of the entries of $I$, i.e., $\supp(I) :=\{I_1,\dots,I_{\ell}\} \subseteq [n]$.
\end{notation}

\begin{proposition} \label{prop:2dimmm}
There are at least 
\begin{align*}
    q!-\sum_{s=\ell+1}^{\min\{2\ell,q\}}\binom{q}{s}\binom{s}{\ell}^2(q-s)!(q-1)q\prod_{i=0}^{s - \ell - 1} (q - i)
\end{align*}
orderings $\balpha= (\alpha_1,\dots,\alpha_q)$ of $\F_q$ such that the code $\mC:=\textup{RS}_{q,2}(\balpha)$ satisfies $\textup{LCS}(\mC) \le \ell-1$, i.e., it can correct at least $q-\ell$ insdel errors.
\end{proposition}

\begin{proof}
     We recall that by Lemma~\ref{lem:alg-cond}, if $\mC:= \textup{RS}_{q,2}(\balpha)$ cannot correct $q-\ell$ insdel errors, then there are two increasing sequences, $I,J\in [q]^{\ell}$ with $d_H(I,J)\geq \ell-1$ such that $\textup{rank} (V_{2,\ell,I,J} (\balpha)) < 3$. Thus, in this proof, we will give an upper bound on the number of distinct orderings $\balpha$ of $\Fq$ for which this (`bad') condition holds.

    Fix $I,J\in[q]^{\ell}$. If $\textup{rank} (V_{2,\ell,I,J} (\balpha)) < 3$, then there exists a nonzero vector $(a,b,-c)\in \Fq^3$ such that $V_{2,\ell,I,J} (\balpha) \cdot (a,b,-c)^T = \bzero$. 
    This implies the following system of $\ell$ equations
    \[
        a \alpha_{I_t} + b = c\alpha_{J_t} \quad \text{for } t = 1, \dots, \ell
    \]
    
    Now, since elements $\alpha_{J_t}, t\in[\ell]$ are pairwise distinct, we can assume that $c\neq 0$ as otherwise, we would get that $a=0$ (since $\alpha_{J_t}, t\in[\ell]$ are also pairwise distinct) and $b=0$.
    Thus, we can assume that $c=1$ and get the following system of $\ell$ equations
    \begin{align} \label{eq:syss}
        a \alpha_{I_t} + b = \alpha_{J_t} \quad \text{for } t \in [\ell]
    \end{align}
    for some $a \in \F_q^*$ and $b \in \F_q$. 
    Let $S:=\supp(I)\cup \supp(J)$ and denote $s := |S|$. Then, the system~\eqref{eq:syss} consists of $\ell$ equations and $s + 2$ unknowns: $\alpha_i$ for $i\in S$ and $a,b$.

    We now give an upper bound for the number of orderings $\balpha$ for which such a system is satisfied. We will need the following claim, that we prove within this proof.

    \begin{claim} \label{cl:unique}
        We can choose $s +2 - \ell$ out of the $s+2$ variables $\{a,b\} \cup \{\alpha_i \mid i\in S\}$ such that by assigning values to these variables, the remaining $\ell$ variables in the system~\eqref{eq:syss} are uniquely determined, if a solution exists.
    \end{claim}

\begin{proof}[Proof of the claim] 
    Note that $|\supp(I) \cap \supp(J)|=2\ell-s$. Hence if $\supp(I) \cap \supp(J) = \emptyset$, then $s = 2\ell$. By fixing $a \in \mathbb{F}_q^*$, $b \in \mathbb{F}_q$, and assigning distinct values from $\mathbb{F}_q$ to $\alpha_i$ with $i \in \supp(I)$, we fully chose all the entries of $\balpha_I$ in this case. Clearly, there are $q\cdot (q-1)\cdot \prod_{i=0}^{\ell - 1} (q - i)$ ways to assign values to these variables.
    Applying the affine transformation $f(x) := ax + b$ to these values then yields all the entries of $\balpha_J$.
    
    Now assume that $\supp(I) \cap \supp(J) \ne \emptyset$ and let $(I_{U_1},\dots,I_{U_u})$ be the increasing sequence made from the elements in $\supp(I) \setminus \supp(J)$ where $u:=s-\ell$. 
    Recall that by Lemma~\ref{lem:alg-cond}, $d_H(I,J)\geq \ell - 1$ and thus, $d_H(I,J)\in \{\ell-1,\ell\}$.

    We begin with the case where $d_H(I,J)=\ell$. Fix values for $s + 2 - \ell$ of the variables in Equation~\ref{eq:syss} in the following way: fix $a \in \F_q^{*}$, $b \in \F_q$ and fix values for $\alpha_i$ for $i \in \supp(I) \setminus \supp(J)$, assigning them distinct elements from~$\F_q$. There are $(q-1)\cdot q\cdot \prod_{i=0}^{s - \ell - 1} (q - i)$ ways to assign values to these variables. Note that in this way, we fixed the values of the variables in $\{\alpha_{I_{U_1}},\dots, \alpha_{I_{U_u}}\}$. 

    Applying~$f(x)=ax+b$ to $\alpha_{I_i}$, $i \in U$, we obtain $$f(\alpha_{I_{U_1}})=\alpha_{J_{U_1}},\dots,f(\alpha_{I_{U_u}})=\alpha_{J_{U_u}}$$ where $(J_{U_1},\dots,J_{U_u})$ is an increasing sequence made from elements in $\supp(J)$.
    Note that there exists some $R_1 \in \{{U_1},\dots,{U_u}\}$ with $J_{R_1} \in \supp(I)$. 
    If this was not the case, then this would mean $\{J_{U_1},\dots,J_{U_u}\} = \supp(J) \setminus \supp(I)$, implying that for the increasing sequence indexed by $M=(M_1,\dots,M_m)$ made from the elements in $\supp(I) \cap \supp(J)$ we have $(I_{M_1},\dots,I_{M_{m}})=(J_{M_1},\dots,J_{M_{m}})$ where $m:=2\ell-s$. This contradicts the fact that $d_H(I,J) = \ell$. 

    Therefore, for this $R_1$, $f(\alpha_{J_{R_1}}) = \alpha_{J_{\tilde{R}_1}}$ for some $J_{\tilde{R}_1} \in \supp(J) \setminus \{J_{U_1},\dots,J_{U_u}\}$, and $J_{R_1} = I_{\tilde{R}_1}$. Now, on top of the (initially chosen) values for $\{\alpha_{I_{U_1}},\dots, \alpha_{I_{U_u}}\}$, we have found the values for $\{\alpha_{J_{U_1}},\dots, \alpha_{J_{U_u}},\alpha_{J_{\tilde{R}_1}}\}$ (which is not necessarily a disjoint set from $\{\alpha_{I_{U_1}},\dots, \alpha_{I_{U_u}}\}$). 

    Now, similarly to before, suppose that $\{J_{U_1},\dots,J_{U_u}, {J_{\tilde{R}_1}}\}\setminus \{J_{R_1}\} = \supp(J) \setminus \supp(I)$. If $|\supp(I) \cap \supp(J)|=1$, we are done. If not, then for the increasing sequence indexed by $M=(M_1,\dots,M_m)$ made from the elements in $(\supp(I) \cap \supp(J)) \setminus \{J_{{R}_1}\}$ we have $(I_{M_1},\dots,I_{M_{m}})=(J_{M_1},\dots,J_{M_{m}})$ where this time $m:=2\ell-s-1$. Again this contradicts that $d_H(I,J) = \ell$.
    Therefore, there exists $R_2 \in \{{U_1},\dots,{U_u},\tilde{R}_1\}$ with $J_{R_2} \in \supp(I)$. Thus, we have found the additional entry of $\alpha$, given by $f(\alpha_{J_{R_2}}) = \alpha_{J_{\tilde{R}_2}}$.

    We can repeat the process above, following the same reasoning, until we find all $\alpha_i$ for $i \in S$ (if a solution exists).

    Now suppose that $d_H(I,J)=\ell-1$. This means that there exists $t \in [\ell]$ with $I_t=J_t$. We fix $a \in \F_q^{*}$, and we fix $\alpha_{I_t}=\alpha_{J_t} \in \F_q$. From this, we automatically find $b$ from the equation $a \alpha_{I_t}+b = \alpha_{I_t}$. Moreover, we fix values for $\alpha_i$ for $i \in \supp(I) \setminus \supp(J)$, assigning them distinct elements from~$\F_q$. There are $(q-1)\prod_{i=0}^{s - \ell} (q - i)$ ways to do this. We then look at the increasing sequences $\tilde{I}$ and $\tilde{J}$ that are obtained from $I$ and $J$, respectively, by shortening the $t$-th coordinate. This gives sequences of length $\ell-1$ with $d_H(\tilde{I},\tilde{J})=\ell-1$, i.e., $\tilde{I}$ and $\tilde{J}$ have the property that in none of their positions their entries coincide. We can then proceed as in the first part of the proof, finding all the values of $\alpha_S$ uniquely (if a solution exists).
\end{proof}

We continue with the proof of the proposition.
Given Claim~\ref{cl:unique}, for a fixed $I,J\in [q]^{\ell}$, there are at most $q\cdot (q-1)\cdot \prod_{i=0}^{s - \ell - 1} (q - i)$ vectors $\balpha_S := (\alpha_i : i\in S) \in \Fq^{|S|}$ for which there exists $a\in \Fq^*, b\in \Fq$ that form a solution to~\eqref{eq:syss}. 

Thus, by taking into consideration that the remaining $q-s$ positions of $\balpha$ (those not in $S$) can then be assigned any distinct values from the unused elements of $\F_q$, we get that there are at most
\[
    (q-s)!\cdot q\cdot (q-1)\cdot \prod_{i=0}^{s - \ell - 1} (q - i)
\]
orderings $\balpha$ of $\Fq$ for which $\textup{rank} (V_{2,\ell,I,J} (\balpha)) < 3$. 

Finally, we observe that the number of ways to choose $I,J\in [q]^{\ell}$ such that $s = |S|= |\supp(I)\cup\supp(J)|$ is at most $\binom{q}{s} \binom{s}{\ell}^2$; $\binom{q}{s}$ choices for $S$, and $\binom{s}{\ell}$ choices for $I$ and $J$ (disregarding the fact that we want $I\ne J$). Summing over all admissible values of $s$ gives that there are at most 
    \[
        \sum_{s = \ell+1}^{\min\{2\ell, q\}} \binom{q}{s} \binom{s}{\ell}^2 (q - s)! (q-1)q\prod_{i=0}^{s - \ell - 1} (q - i)
    \]
number of orderings $\balpha = (\alpha_1,\ldots,\alpha_q)$ of $\Fq$ for which there are two distinct $I,J\in [q]^{\ell}$ with $d_H(I,J)\geq \ell - 1$ such that $V_{2,\ell,I,J}(\balpha)$ does not have full rank. This completes the proof.
\end{proof}

In the following claim, we set $\ell = \delta q$ for some constant $\delta\in[0,1]$ and compute an upper bound on the fraction of bad orderings of $\Fq$, i.e., the orderings that yield a $2$-dimensional RS code that cannot correct $(1 - \delta)q$ insdel errors.


\begin{claim} \label{cl:2dimmm}
Let $0 <\delta < 1$. We have that
\begin{align*}
    \sum_{s=\delta q+1}^{\min\{2\delta q,q\}}\binom{q}{s}\binom{s}{\delta q}^2\frac{(q-s)!}{q!}(q-1)q\prod_{i=0}^{s - \delta q - 1} (q - i)\le q^2  \left( \frac{4e^2}{\delta^2 q} \right)^{\delta q}.
\end{align*}
\end{claim}
\begin{proof}
We have
\begin{align*}
    \sum_{s=\delta q+1}^{\min\{2\delta q,q\}}\binom{q}{s}\binom{s}{\delta q}^2\frac{(q-s)!}{q!}(q-1)q\prod_{i=0}^{s - \delta q - 1} (q - i)
    &\le \sum_{s=\delta q+1}^{2\delta q}\frac{s!}{(\delta q)!^2(s-\delta q)!^2}q^{s-\delta q+2} \\
    &=\sum_{s=1}^{\delta q}\frac{(s+\delta q)!}{(\delta q)!^2s!^2}q^{s+2}
\end{align*}
where the inequality comes from the fact that $\min\{2\delta q, q\} \le 2\delta q$ and 
\begin{align*}
    (q-1)q\prod_{i=0}^{s - \delta q - 1} (q - i) \le q^{s-\delta q+2}.
\end{align*}
For $1 \le s \le \delta q$ let
\begin{align*}
    a_s := \frac{(s+\delta q)!}{(\delta q)!^2s!^2}q^{s+2}.
\end{align*}
Simple simplifications give that for any $1 \le s \le \delta q-1$ we have 
\begin{align*}
    \frac{a_{s+1}}{a_{s}} = \frac{(s+\delta q+1)q}{(s+1)^2}= \frac{q}{s+1} +\frac{\delta q^2}{(s+1)^2} \ge \frac{q}{\delta q }+\frac{\delta q^2}{\delta^2 q^2} = \frac{2}{\delta} > 1.
\end{align*}
Thus, $a_s$ is increasing in $s$, and so the maximal term in the sum is $a_{s}$ for $s=\delta q$. We therefore get the estimate
\begin{align*}
    \sum_{s=1}^{\delta q}\frac{(s+\delta q)!}{(\delta q)!^2s!^2}q^{s+2} \le \delta q \frac{(2\delta q)!}{(\delta q)!^2(\delta q)!^2}q^{\delta q+2} = \delta q \frac{(2\delta q)!}{(\delta q)!^4}q^{\delta q+2}.
\end{align*}

Using Strirling's approximation~\cite[Lemma 7.3]{MU} which states that for any integer $m$, we have that
\[
    \sqrt{2\pi m} \left(\frac{m}{e}\right)^m \leq m! \leq 2\sqrt{2\pi m} \left(\frac{m}{e}\right)^m \;,
\]
we get
\begin{align*}
    (2\delta q)! \leq 4 \sqrt{\pi \delta q}\left(\frac{2\delta q}{e}\right)^{2\delta q} \quad \textup{and} \quad
    (\delta q)!^4 \geq 4\pi^2 \delta^2 q^2 \left(\frac{\delta q}{e}\right)^{4\delta q}\;.
\end{align*}
Therefore,
\begin{align*}
    \delta q \frac{(2\delta q)!}{(\delta q)!^4}q^{\delta q+2} &\leq \delta q \frac{\sqrt{\pi \delta q}\left(\frac{2\delta q}{e}\right)^{2\delta q}}{\pi^2 \delta^2 q^2 \left(\frac{\delta q}{e}\right)^{4\delta q}}  q^{\delta q + 2} 
    =\frac{1}{\sqrt{\pi^3\delta q}} \left( \frac{2e}{\delta q} \right)^{2\delta q}  q^{\delta q + 2} \leq q^2  \left( \frac{4e^2}{\delta^2 q} \right)^{\delta q}\;
\end{align*} 
which proves the claim.
\end{proof}

Combining Proposition~\ref{prop:2dimmm} with the upper bound given in Claim~\ref{cl:2dimmm}, we prove Theorem~\ref{thm:2-dim}, which is restated for convenience. 

\twodim*

\begin{proof}
    By Proposition~\ref{prop:2dimmm}, at least 
    \begin{equation} \label{eq:num-of-good-ord}
        q!-\sum_{s=\ell+1}^{\min\{2\ell,q\}}\binom{q}{s}\binom{s}{\ell}^2(q-s)!(q-1)q\prod_{i=0}^{s - \ell - 1} (q - i)
    \end{equation}
    of orderings $\balpha$ of $\Fq$, yield a $2$-dimensional RS code that can correct $q-\ell$ insdel errors.

    Now, set $\ell = \delta q$. By Claim~\ref{cl:2dimmm}, we know that~\eqref{eq:num-of-good-ord} is at least
    \[
    \left( 1 - q^2 \left( \frac{4e^2}{\delta^2 q} \right)^{\delta q} \right) q!\;.
    \]
    Note that for every $q>8e^2/\delta^2$, we have that $q^2 \left( \frac{4e^2}{\delta^2 q} \right)^{\delta q} \leq 2^{-\delta q + 2\log q}\leq \exp(-\Omega(\delta q))$. Thus, for every $\varepsilon$, there exists a large enough prime power $q$ such that $\exp(-\Omega(\delta q))<\varepsilon$. The theorem follows.
\end{proof}

Note that the previous result shows that for any $0< \delta < 1$, if one picks uniformly at random a full-length 2-dimensional Reed-Solomon code over $\F_q$ 
from the set of all possible such codes, then with high probability this code will be able to correct $(1-\delta)q$ insdel errors, as long as $q$ is large enough.

\subsection{General $k$: Existence of Full-length RS Coded Correcting Insdel Errors.}
Our next goal is to use the power of randomness and show that there is a positive probability that a random permutation of $\F_q$ (when $q$ is large enough compared to~$k$) will give rise to a RS code that can correct many insdel errors. 
In the following, we prove Theorem~\ref{thm:main} which is restated for convenience.
\existperm*
We will need the following claim in the proof.

\begin{claim} \label{clm:distinct-alphas} 
    Let $\balpha = (\alpha_1, \ldots, \alpha_q) \in \F_q^q$ be a uniformly random vector. Then, the probability that all distinct $i,j \in [q]$, we have $\alpha_i \ne \alpha_j$ is at least $e^{-q}$.
\end{claim}
\begin{proof}
    The number of permutations on $q$ elements is $q!$. Thus, this probability is exactly $q!/q^q$. By Stirling's approximation (e.g., \cite[Lemma 7.3]{MU}), we get 
    \[
    \frac{q!}{q^q} \geq \frac{\sqrt{8\pi q} \left(\frac{q}{e}\right)^q}{q^q} \geq e^{-q}\;. \qedhere
    \]
\end{proof}
\begin{proof}[Proof of Theorem~\ref{thm:main}]
    Let $\balpha = (\alpha_1, \ldots, \alpha_q) \in \F_q^q$ be a uniform random vector. Note that it can be that $\balpha$ is not a permutation of $\Fq$ and we will address this issue at the end of the proof. Set $\ell = q - q/(16 k)$. For two increasing sequences $I,J \subset [q]^{\ell}$ that agree on at most $k-1$ coordinates, we will give an upper bound on the probability that $V_{k,\ell,I,J}(\balpha)$ is not of full rank.

    For this purpose, consider the matrix $V_{k,\ell,I,J}(\bX)$ from \eqref{eq:mat-lcs-eq} where $\bX = (X_1, \ldots, X_n)$ are formal variables. Consider the matrix $M(\bX)$ obtained by taking the top $(2k - 1)\times (2k-1)$ submatrix of $V_{k,\ell,I,J}(\bX)$. Its determinant is a non-zero polynomial \cite[Proposition 18]{con2023reed} of degree less than $k^2$. Thus, by the Schwartz-Zippel Lemma \cite{Zippel79,Schwartz80}, $\Pr[\det(M(\balpha)) = 0] \leq k^2/q$ where the probability is over all $\balpha \in \F_q^q$. 

    We shall construct a sequence $M_1(\bX), \ldots, M_{m}(\bX)$ of $(2k-1) \times (2k - 1)$ submatrices of $V_{k,\ell,I,J}(\bX)$. Denote by $\textup{Vars}(M_j(\bX)) = \{ i\in [n] : X_i \textup{ appears in } M_j(\bX)\}$.
    We require our sequence of submatrices to comply with the following condition: for all distinct $i,j\in [m]$ it holds that
    \[
    \textup{Vars}(M_i(\bX)) \cap \textup{Vars}(M_j(\bX)) = \emptyset \;.
    \]
    In other words, this condition ensures that for any distinct $i,j\in [m]$, the set of variables that appear in $M_i$ is disjoint from the set of variables that appear in $M_j$.
    This implies that the event $\det(M_i(\balpha)) = 0$ is independent of all the events $\{\det(M_j(\balpha))=0:  j\in [m]\setminus \{i\} \}$ for all $i \in [m]$. 
    Therefore, 
    \begin{align*}
        &\Pr_{\balpha\sim \F_q^q}[\textup{rank}(V_{k,\ell, I,J}(\balpha)) < 2k-1] \\
        &\;\;\;  \leq \Pr_{\balpha\sim \F_q^q}[\forall i\in [m],\; \det(M_i(\balpha)) = 0] \\
        &\;\;\; \leq \prod_{i\in [m]} \Pr_{\balpha\sim \F_q^q} [\det(M_i(\balpha) = 0)]  \leq \left( \frac{k^2}{q} \right)^m \;,
    \end{align*}
    where the first inequality is by observing that if $\textup{rank}(V_{k,\ell, I,J}(\balpha))(\balpha) < 2k-1$ then each $(2k-1)\times (2k-1)$ submatrix must be singular and because of the independence argument above, the probabilities for each matrix can be multiplied.

    Next, we claim that $m$ is at least $\ell/4k$. Choose $M_1$ by picking the first $2k-1$ rows out of the $\ell$ rows of $V_{k,\ell,I,J}(\bX)$. Now, there are $\ell -(2k - 1)$ rows from which we can choose $M_2$. 
    Our next observation is that there are at most $2k-1$ rows out of these $\ell -(2k - 1)$ rows that contain a variable which is also in $\textup{Vars}(M_1(\bX))$. Indeed, consider the $s$-th row in $M_1(\bX)$ which contains the variable $X_i$ and $X_j$ and assume w.l.o.g. that $i \leq j$. For every $s'>s$, we cannot have that the $s'$-th row in $V_{k,\ell,I,J}(\bX)$ will contain the variable $X_i$ since it would contradict the fact that $I,J$ are increasing sequences.
    Therefore, for any row that was added to $M_1(\bX)$, only (at most) one row cannot be added to $M_2(\bX)$. Thus, there are $\ell - (4k - 2)$ rows to choose from when constructing $M_2(\bX)$. Choose the first $2k-1$ rows and continue this way.
    We conclude that the number of matrices $m$ is at least $\ell/4k$.

    Thus, for fixed increasing sequences $I,J$ of length $\ell$ that agree on at most $k-1$ coordinates, the probability that $V_{k,\ell,I,J}(\balpha)$ is not full rank is at most 
    \[
    \left( \frac{k^2}{q} \right)^{\ell/(4k)}\;.
    \]
    By the union bound, the probability that there exists a pair $I,J$ of increasing sequences of length $\ell$ that agree on at most $k-1$ coordinates for which the rank of $V_{k,\ell,I,J}(\balpha)$ is not full, is at most
    \begin{align*}
        \left( \frac{k^2}{q} \right)^{\ell/(4k)} \binom{q}{\ell}^2 
        &\leq \left( \frac{k^2}{q} \right)^{q/(5k)} \left(\frac{e q}{\frac{q}{10k}}\right)^{q/(5k)} \\
        & = \left( \frac{10ek^3}{q} \right)^{q/(5k)} \leq e^{-\frac{6}{5} q},
    \end{align*}
    where the first inequality is by our requirement that the code can correct $q/(10k)$ insdel errors, which implies that $\ell = q - q/(10k)$ and it holds that $\ell/(4k) > q/(5k)$. Furthermore, we used the inequality $\binom{n}{k} < (en/k)^k$. 
    The final inequality follows by our assumption that $q \geq e^{6k} \cdot (10ek^3)$.
    
    Finally, observe that in order for our randomly chosen $\balpha$ to define a proper RS code, it must be that all points in $\balpha$ are distinct. 
    From Claim~\ref{clm:distinct-alphas} we know that the probability that this happens is at least $e^{-q}$. 
    Thus, a random vector $\balpha\in \Fq^q$ does not give rise to an RS code that can correct $q/(10k)$ insdel errors if $\balpha$ is not a permutation or it is a permutation but there exists a pair of increasing sequences $I,J\subset[n]^{q - q/(10k)}$ that agree on at most $k-1$ coordinates such that $\det(V_{k, q - q/(10k),I,J}(\balpha)) = 0$.
    This occurs with probability at most $(1 - e^{-q}) + e^{-6q/5} < 1$. 
    We conclude that there exists an $\balpha \in \Fq^n$ such that the respective $\RSfull$ can correct $q/(10k)$ insdel errors. 
\end{proof}

\section{Rate-$1/2$ RS codes correcting a single insdel error.} \label{sec:ex}

In this section, we give new existence results for optimal RS codes. Our approach is based on the following (easy) lemma, which can be seen as a generalization of~\cite[Lemma 4.10]{Xing} for general $k$.

\begin{lemma} \label{lem:iff}
A $\RSS$ code is optimal if and only if $f(\balpha_I) \ne {g}(\balpha_J)$ for any $f \in \{f_{k-1}x^{k-1}+\dots +f_1x : f_i \in \F_q \textnormal{ for all $i \in \{1,\dots,k-2\}$}, f_{k-1} \in \{0,1\}\}$ and $g \in \F_q[x]_{< k}$, with $f\ne g$ and increasing sequences $I, J \in [n]^{2k-1}$.
\end{lemma}
\begin{proof}
Suppose there exist $f(x) \in \{f_{k-1}x^{k-1}+\dots +f_1x : f_i \in \F_q \textnormal{ for all $i \in \{1,\dots,k-2\}$},$ $f_{k-1} \in \{0,1\}\}$ and $g(x) \in \F_q[x]_{< k}$ with $f \ne g$ and increasing vectors $I, J \in [n]^{2k-1}$ such that $I^{f}_\balpha = J^{g}_\balpha$. Since $f \ne g$, $f(\balpha)\ne g(\balpha)$ and $\textup{LCS}(f(\balpha),g(\balpha)) \ge 2k-1$ because they at least share the subsequence indexed  by $I$ and $J$, respectively. Thus the code is not optimal.

Now suppose that the code $\RSS$ satisfies $\textup{LCS}(\RSS) \ge 2k-1$. Then there exist two distinct codewords $\bc,\bc' \in \textup{RS}_{n,k}(\balpha)$
with $\textup{LCS}(\bc,\bc') \ge 2k-1$. W.l.o.g. we assume $\textup{LCS}(\bc,\bc') = 2k-1$. We can write $\bc = f(\balpha)$ and $\bc' = g(\balpha)$ for some distinct $f, g \in \F_q[x]_{< k}$. Denote the indices of the largest common subsequence of $\bc$ and $\bc'$ (i.e. the coordinates where they coincide) by $I$ and $J$, respectively. We have $I,J \in [n]^{2k-1}$ and $f(\balpha_I) = g(\balpha_J)$. 
By the equations in~\eqref{eq:isometrie1} and~\eqref{eq:isometrie2}, we can subtract and multiply by suitable elements of~$\F_q$ to obtain $\textup{LCS}(\bc,\bc')=\textup{LCS}(f(\balpha),g(\balpha)) = \textup{LCS}(\tilde{f}(\balpha),\tilde{g}(\balpha)) \ge 2k-1$ where $\tilde{f}(x) \in \{\tilde{f}_{k-1}x^{k-1}+\dots +\tilde{f}_1x : \tilde{f}_i \in \F_q \textnormal{ for all $i \in \{1,\dots,k-2\}$}, \tilde{f}_{k-1} \in \{0,1\}\}$ and $\tilde{g} \in \F_q[x]_{<k}$.
\end{proof}

\subsection{The Dimension $2$ Case} \label{subsec:2}
Already for $k=2$ it is an open problem to determine the smallest $q$ for which there exists an optimal RS code $\text{RS}_{n,2}(\balpha)$.  
In this subsection we  prove that $\mathbb{F}_7$ is the smallest finite field for which optimal RS codes with these parameters exist. In Subsection~\ref{subsec:k}, we build on the result of this subsection and establish the existence of optimal RS codes with rate $1/2$ through induction on $k \geq 2$.  

%
%

\begin{lemma} \label{lem:k2}
 The smallest $q$ for which there exists an optimal $\text{RS}_{4,2}(\balpha)$ code over $\mathbb{F}_q$ is $q=7$.  In this case, a possible evaluation vector is $\balpha=(0,1,2,5) \in \F_7^4$.
\end{lemma}
\begin{proof}
Because of Lemma \ref{lemxing}, the capability of $\balpha$ giving rise to an optimal $\text{RS}_{4,2}(\balpha)$-code over $\mathbb{F}_q$ depends on the action of $AGL(\mathbb{F}_q)$, which is $2$-transitive on $\mathbb{F}_q$, see \cite[Lemma 2.2 (i)]{Xing}. This allows to choose the first two coordinates in $\balpha$ arbitrarily, as long as they are distinct, and so w.l.o.g. we can assume that $\balpha=(0,1,\alpha_1,\alpha_2)$ with $\alpha_1,\alpha_2 \ne 0,1$ and $\alpha_1 \ne \alpha_2$. 
From Lemma~\ref{lemxing} the code $\text{RS}_{4,2}(\balpha)$ has insdel distance $2n-4$ if and only if there is no non-trivial $f_{a,b} \in AGL(\mathbb{F}_q)$ that maps any triple in the following set to any other triple in the same set: $S:=\{(0,1,\alpha_1),(0,1,\alpha_2), (0,\alpha_1,\alpha_2), (1,\alpha_1,\alpha_2)\}$. 
The strategy now is to compute all the images through $f_{a,b}$ of all the triples in $S$, and force the condition described by the if and only if to be satisfied. Clearly we will assume that $(a,b) \ne (1,0)$, as otherwise $f_{a,b}$ would be the identity map.
\begin{itemize}
    \item[(i)] $f_{a,b}(0,1,\alpha_1)=(b,a+b,a\alpha_1+b) \in S$ forces either $b=0$ and $a=\alpha_1$ (note that if $a=1$ then $f_{a,b}$ is just the identity map) or $b=1$ and $a=\alpha_1-1$. In the first case we get that $f_{a,b}(0,1,\alpha_1)=(0,\alpha_1,\alpha_1^2)  \in S$ if and only if $\alpha_2= \alpha_1^2$. In the latter case $f_{a,b}(0,1,\alpha_1)=(1,\alpha_1,\alpha_1^2-\alpha_1+1)  \in S$ if and only if $\alpha_2= \alpha_1^2-\alpha_1+1$. In conclusion for all $(a,b) \ne (0,1)$ we have that $f_{a,b}(0,1,\alpha_1)  \not\in S$ exactly when $\alpha_2 \ne \alpha_1^2,\alpha_1^2-\alpha_1+1$.
    
    \item[(ii)] $f_{a,b}(0,1,\alpha_2)=(b,a+b,a\alpha_2+b) \in S$. As before this condition forces either $b=0$ and $a=\alpha_1$ or $b=1$ and $a=\alpha_1-1$. In the first case we get that $f_{a,b}(0,1,\alpha_2)=(0,\alpha_1,\alpha_1\alpha_2)  \in S$ if and only if $\alpha_1= 1$, which clearly cannot happen. In the latter case $f_{a,b}(0,1,\alpha_1)=(1,\alpha_1,\alpha_2(\alpha_1-1)+1)  \in S$ if and only if $\alpha_2=\alpha_2(\alpha_1-1)+1$. In conclusion, for all $(a,b) \ne (0,1)$ we have that $f_{a,b}(0,1,\alpha_2)  \not\in S$ when $\alpha_2(\alpha_1-2) \ne -1$. 
    
    \item[(iii)] $f_{a,b}(0,\alpha_1,\alpha_2)=(b,a\alpha_1+b,a\alpha_2+b) \in S$ forces either $b=0$ and $a=\alpha_1^{-1}$ or $b=1$ and $a\alpha_1+1=\alpha_1$. In the first case we get $f_{a,b}(0,\alpha_1,\alpha_2)=(0,1,\alpha_2/\alpha_1) \in S$ if and only if $\alpha_2/\alpha_1=\alpha_1$ or $\alpha_2/\alpha_1=\alpha_2$. This can happen only if $\alpha_2=\alpha_1^2$ as $\alpha_1 \ne 1$. In the second case we get $f_{a,b}(1,\alpha_1,\alpha_2)=(1,\alpha_1,\alpha_2(\alpha_1-1)/\alpha_1+1) \in S$ if and only if $\alpha_2(\alpha_1-1)/\alpha_1+1=\alpha_2$. This case never happens as $\alpha_1 \ne \alpha_2$. In conclusion, for all $(a,b) \ne (0,1)$ we have that $f_{a,b}(0,\alpha_1,\alpha_2)=(b,a\alpha_1+b,a\alpha_2+b) \not\in S$ when $\alpha_2 \ne \alpha_1^2$. 
    
    \item[(iv)] $f_{a,b}(1,\alpha_1,\alpha_2)=(a+b,a\alpha_1+b,a\alpha_2+b) \in S$ forces $a+b=0$ and either $a\alpha_1-a=1$ or $a\alpha_1-a=\alpha_1$. This is because this time if $a+b=1$ then we need $a\alpha_1+b=a\alpha_1+1-a=\alpha_1$ and hence $a=1$ (recall that $\alpha_1 \ne 1$). As before, if $(a,b)=(1,0)$ then we just have the identity map and we discard this case. If $a+b=0$ and $a\alpha_1-a=1$ then we get $f_{a,b}(1,\alpha_1,\alpha_2)=(0,1,(\alpha_2-1)/(\alpha_1-1)) \in S$ only if either $\alpha_2=\alpha_1^2-\alpha_1+1$ or $\alpha_2(\alpha_1-2)=-1$.
    If $a+b=0$ and $a\alpha_1-a=\alpha_1$ then $f_{a,b}(1,\alpha_1,\alpha_2)=(0,\alpha_1,(\alpha_2-1)\alpha_1/(\alpha_1-1))$ which is never in $S$ as $\alpha_1 \ne \alpha_2$.
    In conclusion for all $(a,b) \ne (1,0)$, one has $f_{a,b}(1,\alpha_1,\alpha_2)=(a+b,a\alpha_1+b,a\alpha_2+b) \not\in S$ when $\alpha_2 \ne \alpha_1^2-\alpha_1+1$ and $\alpha_2(\alpha_1-2) \ne -1$.
    
\end{itemize}

Summarizing up all the cases above we get that $\text{RS}_{4,2}(\balpha)$ with $\balpha=(0,1,\alpha_1,\alpha_2)$, $\alpha_1 \ne 0,1$ and $\alpha_2 \ne 0,1,\alpha_1$ has insdel distance $2n-4$ if and only if $\alpha_2 \not\in \{0,1,\alpha_1,\alpha_1^2,\alpha_1^2-\alpha_1+1\}$ and $\alpha_2 \ne -1/(\alpha_1-2)$ for $\alpha_1 \ne 2$.
This means $q \geq 7$ is necessary to find a good $\balpha$, as if $q=4,5$ then there is no $\alpha_2 \in \mathbb{F}_q$ that can satisfy the conditions above for any choice of $\alpha_1$.
On the other hand if $q=7$ then $\balpha=(0,1,2,5)$ satisfies all the required conditions.
\end{proof}

The following is a consequence of the proof of the previous lemma.

\begin{remark} \label{rem:rs-2-4}
    The $\text{RS}_{4,2}(\balpha)$ code with $\balpha=(0,1,\alpha_1,\alpha_2)$, $\alpha_1 \ne 0,1$ and $\alpha_2 \ne 0,1,\alpha_1$ is optimal if and only if $\alpha_2 \not\in \{0,1,\alpha_1,\alpha_1^2,\alpha_1^2-\alpha_1+1\}$ and $\alpha_2 \ne -1/(\alpha_1-2)$ for $\alpha_1 \ne 2$.
\end{remark}

\subsection{Induction on $k$: Rate $1/2$}
\label{subsec:k}

In this subsection we use the result of the previous subsection as the base case, and we apply induction on $k$ for rate $1/2$ RS codes. We need two claims to prove the main statement which is stated in Proposition~\ref{prop:exist}. 

\begin{claim} \label{cl:intersections}
    Let $\ell \ge 2$ and $I,J \in [\ell]^{\ell-1}$ be increasing sequences. Then if $I \ne J$ we have $d_H(I,J) = |s_J-s_I|$ where $s_I$ is the unique element in $\{1,\dots,\ell\} \setminus \{I_1,\dots, I_{(\ell-1)}\}$ and $s_J$ is the unique element in $\{1,\dots,\ell\} \setminus \{J_1,\dots, J_{(\ell-1)}\}$.
\end{claim}
\begin{proof}
    Suppose $s_I=s_J$. Then we clearly have $I=J$, contradiction. Therefore we can assume w.l.o.g. that $s_I < s_J$. We have
    \begin{align} \label{eq:iuu}
        I_u = \begin{cases}
            u &\textnormal{ if $1 \le u < s_I$} \\
            u+1 &\textnormal{ if $s_I \le u \le \ell-1$} \\
        \end{cases}
    \end{align}
    \begin{align*}
           J_u = \begin{cases} 
            u &\textnormal{ if $1 \le u < s_J$} \\
            u+1 &\textnormal{ if $s_J \le u \le \ell-1$}. \\
        \end{cases} 
    \end{align*}
    Therefore, we obtain that 
    \begin{align*}
        d_H(I,J) &= |\{u \in \{1,\dots,\ell-1\}: I_u \ne J_u\}| \\
        &= \ell-1-|\{u \in \{1,\dots,\ell-1\}: I_u = J_u\}| \\
        &= \ell-1-(\min\{s_I,s_J\}-1+\ell-\max\{s_I,s_J\}) \\
        &= \ell-1-s_I+1-\ell+s_J \\
        &=s_J-s_I,
    \end{align*}
    which is the statement of the lemma.
\end{proof}

\begin{claim} \label{cl:22}
Let $\ell \ge 2$, let $\balpha=(\alpha_1,\dots,\alpha_{\ell}) \in \F_q^{\ell}$ be a vector of pairwise distinct elements of $\F_q$ and let $f \in \F_q[x]_{<k}$. If for increasing sequences $I, J \in [\ell]^{\ell-1}$ with $d_H(I,J) \ge k-1$ we have $f(\balpha_I)=f(\balpha_J)$, then $f$ is constant.
\end{claim}
\begin{proof}
Let $s_I$ and $s_J$ be defined as in Claim~\ref{cl:intersections}. W.l.o.g. assume that $s_I < s_J$. Since $f(\balpha_I)=f(\balpha_J)$ we have that $$(f(\alpha_{I_1}),\dots, f(\alpha_{I_{\ell-1}})) = (f(\alpha_{J_1}),\dots, f(\alpha_{J_{\ell-1}}))$$ and by the Equations~\eqref{eq:iuu} from the proof of Claim~\ref{cl:intersections} we have
\begin{align*}
        \begin{cases}
       f(\alpha_{s_I+1})&= f(\alpha_{s_I}), \\
        f(\alpha_{s_I+2}) &= f(\alpha_{s_I+1}), \\
                              &\vdots                  \\
        f(\alpha_{s_I+(s_J-s_I)}) &= f(\alpha_{s_I+(s_J-s_I)-1}).
    \end{cases}
\end{align*}
This means in particular that
\begin{align*}
    f(\alpha_{s_I}) = f(\alpha_{s_I+1}) = \dots = f(\alpha_{s_I+(s_J-s_I)}),
\end{align*}
where $s_J-s_I+1 = d_H(I,J) +1\ge k$. Then if $z:=f(\alpha_{s_I})$, this means that $f(x)-z$ has at least $s_J-s_I+1\ge k$ zeros (i.e. $\{\alpha_{s_I},\dots,\alpha_{s_I+(s_J-s_I)}\}$ are zeros of $f(x)-z$) and it has degree at most $k-1$. Therefore $f(x)=z$ and so $f$ is constant.
\end{proof}

\begin{proposition} \label{prop:exist}
Let $k \ge 2$. If $q \ge 20k^4 - 90k^3 + 150k^2 - 106k + 27$ then there exists an evaluation vector~$\balpha \in \F_q^{2k}$ such that the code $\text{RS}_{2k,k}(\balpha)$ is optimal.
\end{proposition}
\begin{proof}
We prove the statement by induction on $k$.

For $k=2$ from Lemma~\ref{lem:iff} we know that for $q \ge 7$ we can find an evaluation vector that has the desired properties.

Now assume that we have found a vector $\widetilde{\balpha} = (\alpha_1,\dots,\alpha_{2k-2}) \in \F_q^{2k-2}$ for which the code $\text{RS}_{2k-2,k-1}(\widetilde{\balpha})$ has optimal insdel distance 4, and thus by Lemma~\ref{lem:iff} has the property that ${f(\widetilde{\balpha}_I)} \ne {g(\widetilde{\balpha}_J)}$ for any $f \in \{f_{k-2}x^{k-2}+\dots +f_1x : f_i \in \F_q \textnormal{ for all $i \in \{1,\dots,k-3\}$}, f_{k-2} \in \{0,1\}\}$ and $g \in \F_q[x]_{< k-1}$, with $f\ne g$ and any increasing sequences $I, J \in [2k-2]^{2k-3}$.
In order to obtain a $k$-dimensional code with insdel distance $4$ in $\F_q$, again by Lemma~\ref{lem:iff} we need to make sure that we have ${f}(\balpha_I) \ne {g}(\balpha_J)$ for any $f \in \{f_{k-1}x^{k-1}+\dots +f_1x : f_i \in \F_q \textnormal{ for all $i \in \{1,\dots,k-2\}$}, f_{k-2} \in \{0,1\}\}$ and $g \in \F_q[x]_{< k}$ with $f \ne g$ and any increasing sequences $I, J \in [2k]^{2k-1}$. We count the tuples $(\alpha_{2k-1},\alpha_{2k}) \in \F_q^2$ that we need to exclude in order to make sure that the evaluation points $(\alpha_1,\dots, \alpha_{2k})$ give a RS code whose largest common subsequence is of length $2k-1$. 

Let  $I=(I_1,\dots,I_{2k-1}), J=(J_1,\dots,J_{2k-1}) \in [2k]^{2k-1}$ be increasing sequences, and denote $I^\star:=(I_1,\dots,I_{2k-3})$ and $J^\star=(J_1,\dots,J_{2k-3})$. We clearly have $I^\star,J^\star \in [2k-2]^{2k-3}$ and there are a total of $\binom{2k-2}{2}=(2k-2)(2k-3)/2$ choices for $I^\star$ and $J^\star$ (such that they are not the same). Note that if $f({\balpha}_I) = g({\balpha}_J)$ then also $f({\balpha}_{I^{\star}}) = g({\balpha}_{J^{\star}})$ for some $f,g \in \F_q[x]_{< k}$ with $f\ne g$. From imposing that $f({\balpha}_{I^{\star}}) = g({\balpha}_{J^{\star}})$, we obtain the following system of $2k-3$ linear equations (in the variables $f_0,\dots,f_{k-1},g_0,\dots,g_{k-1}$):
    \begin{align*}
        \begin{cases}
            f_{k-1}\alpha_{I_1}^{k-1}+\dots +f_0 &= g_{k-1}\alpha_{J_1}^{k-1}+\dots+g_0 \\
            f_{k-1}\alpha_{I_2}^{k-1}+\dots +f_0 &= g_{k-1}\alpha_{J_2}^{k-1}+\dots+g_0 \\
            &\vdots \\
             f_{k-1}\alpha_{I_{2k-3}}^{k-1}+\dots +f_0 &= g_{k-1}\alpha_{J_{2k-3}}^{k-1}+\dots+g_0.
        \end{cases}
    \end{align*}
    Now suppose that $f_{k-1}=g_{k-1}=0$. Then $f, g \in \F_q[x]_{< k-1}$, $f \ne g$, and we have $f({\balpha}_{I^{\star}}) = g({\balpha}_{J^{\star}})$, contradicting the fact that $\widetilde{\balpha}=(\alpha_1,\dots,\alpha_{2k-3})$ was chosen such that $\text{RS}_{2k-2,k-1}(\widetilde{\balpha})$ has optimal insdel distance 4. Therefore at least one of $f_{k-1}$ or $g_{k-1}$ has to be non-zero. W.l.o.g. assume $f_{k-1}\ne 0$, and we rewrite the set of equations as follows (subtracting from both sides $f_0$, dividing both sides by $f_{k-1}$, and rearranging):
    \begin{align} \label{eq:linear-system-for-alg}
        \begin{cases} \sum_{j=0}^{k-2}\widetilde{g}_j\alpha_{J_1}^j - \sum_{i=1}^{k-2}\widetilde{f}_i\alpha_{J_1}^i &=  \alpha_{I_1}^{k-1}-\widetilde{g}_{k-1}\alpha_{J_1}^{k-1} \\
        \sum_{j=0}^{k-2}\widetilde{g}_j\alpha_{J_2}^j - \sum_{i=1}^{k-2}\widetilde{f}_i\alpha_{J_2}^i &=  \alpha_{I_2}^{k-1}-\widetilde{g}_{k-1}\alpha_{J_2}^{k-1} \\
            &\vdots \\
        \sum_{j=0}^{k-2}\widetilde{g}_j\alpha_{J_{2k-3}}^j - \sum_{i=1}^{k-2}\widetilde{f}_i\alpha_{J_{2k-3}}^i &=  \alpha_{I_{2k-3}}^{k-1}-\widetilde{g}_{k-1}\alpha_{J_{2k-3}}^{k-1} \\
        \end{cases}
    \end{align}
    where $\widetilde{g}(x) := \frac{g(x)-f_0}{f_{k-1}}$ and $\widetilde{f}(x) := \frac{f(x)-f_0}{f_{k-1}}$).
    This system of $2k-3$ linear equations has $2k-2$ unknowns, and so we assign a fixed value in $\F_q$ to $\widetilde{g}_{k-1}$. We want to show now that there is at most one solution to this system of equations, and so $\widetilde{f}$ and~$\widetilde{g}$ are uniquely determined (up to choosing $\widetilde{g}_{k-1} \in \F_q$ freely). In order to do so, it is enough to show that the kernel of the following matrix is trivial:
    \begin{align*}
        M:=\begin{pmatrix} 
	1 & \alpha_{J_1} & \ldots & \alpha_{J_1}^{k-2}  & \alpha_{I_1} &\ldots & \alpha_{I_1}^{k-2} \\ 
	1 & \alpha_{J_2} & \ldots & \alpha_{J_2}^{k-2}  & \alpha_{I_2} &\ldots & \alpha_{I_2}^{k-2} \\
	\vdots &\vdots & \ldots &\vdots &\vdots &\ldots &\vdots \\
	1 & \alpha_{J_{2k-3}} & \ldots & \alpha_{J_{2k-3}}^{k-2}  & \alpha_{I_{2k-3}} &\ldots & \alpha_{I_{2k-3}}^{k-2}\\
	\end{pmatrix} .
    \end{align*}
    From Lemma~\ref{lem:alg-cond} we know that the only possible vectors in the kernel of $M$ are of the form $(0,h_1,\dots,h_{k-2},-h_1,\dots,-h_{k-2})$ for some $h_1,\dots,h_{k-2} \in \F_q$. Let $$(0,h_1,\dots,h_{k-2},-h_1,\dots,-h_{k-2})$$ be in the kernel of $M$ and define $h(x) = h_1x+h_2x^2+\dots + h_{k-2}x^{k-2} \in \F[x]_{< k-1}$. We have $h({\balpha}_{I^\star}) = h({\balpha}_{J^\star})$ and $h({\balpha}_I)=h({\balpha}_{J})$. Note also, since $I,J \in [2k]^{2k-1}$ with $d_H(I,J) \ge k$, we have $d_H(I^\star,J^\star) \ge k-2$. Therefore, by Claim~\ref{cl:22} we obtain that $h$ is constant, and since $h(x) = h_1x+h_2x^2+\dots + h_{k-2}x^{k-2}$, this implies $h=0$.

    For $I,J \in [2k]^{2k-1}$ we have $2k-2 \le I_{2k-2} < I_{2k-1} \le 2k$ and $2k-2 \le J_{2k-2} < J_{2k-1} \le 2k$. There are a total of $\binom{3}{2}^2$ options for $(I_{2k-2},I_{2k-1},J_{2k-2}, J_{2k-1})$. 
    However, we cannot have $I_{2k-2}=J_{2k-2}=2k-2$, because this would force $(I_1,\dots,I_{2k-2}) = (J_1,\dots,J_{2k-2})$ and since $d_H(I,J) \ge k$ this cannot hold. Therefore, there are $\binom{3}{2}^2-\binom{2}{1}^2=5$ possible options for $(I_{2k-2},I_{2k-1},J_{2k-2}, J_{2k-1})$. 
    For each of the possible realizations of $(I_{2k-2},I_{2k-1},J_{2k-2}, J_{2k-1})$ we obtain a set of equations in $\alpha_{2k-1},\alpha_{2k}$, and so in total we have the following 5 sets of equations in $\alpha_{2k-1},\alpha_{2k}$ that would lead to problems:
    \begin{align}\label{eq:bad-pairs-computations}
    \begin{cases}
        &\widetilde{f}(\alpha_{2k-2})=\widetilde{g}(\alpha_{2k-1}) \\
        &\widetilde{f}(\alpha_{2k-1})=\widetilde{g}(\alpha_{2k})
    \end{cases}
    \begin{cases}
        &\widetilde{f}(\alpha_{2k-2})=\widetilde{g}(\alpha_{2k-1}) \\
        &\widetilde{f}(\alpha_{2k})=\widetilde{g}(\alpha_{2k})
    \end{cases}
    \begin{cases}
        &\widetilde{f}(\alpha_{2k-1})=\widetilde{g}(\alpha_{2k-2}) \\
        &\widetilde{f}(\alpha_{2k})=\widetilde{g}(\alpha_{2k-1})
    \end{cases} 
     \\
    \begin{cases}
        &\widetilde{f}(\alpha_{2k-1})=\widetilde{g}(\alpha_{2k-2}) \\
        &\widetilde{f}(\alpha_{2k})=\widetilde{g}(\alpha_{2k})
    \end{cases}
    \begin{cases}
        &\widetilde{f}(\alpha_{2k-1})=\widetilde{g}(\alpha_{2k-1}) \\
        &\widetilde{f}(\alpha_{2k})=\widetilde{g}(\alpha_{2k}).
    \end{cases}\nonumber
    \end{align}
    Each of these systems has at most $(k-1)^2$ solutions for $(\alpha_{2k-1},\alpha_{2k})$ (because $\widetilde{f}(x), \widetilde{g}(x) \in \F_q[x]_{<k}$). Since we chose $\widetilde{g}_{k-1} \in \F_q$ freely, there are at most $5(k-1)^2q$ tuples $(\alpha_{2k-1},\alpha_{2k})$ that would cause problems. We also need to impose that $\alpha_{2k-1},\alpha_{2k} \notin \{\alpha_1,\dots,\alpha_{2k-2}\}$. Let $\mA:=\{\alpha_1,\dots,\alpha_{2k-2}\}$. Therefore we have that whenever
    \begin{align*}
        \binom{|\F_q\setminus \mA|}{2} \ge (2k-2)(2k-3)5(k-1)^2q/2
    \end{align*}
    then we have enough elements in $\F_q$ to choose $(\alpha_{2k-1},\alpha_{2k}) \in \F_q^2$ that neither satisfy the 5 set of equations, nor already show up as evaluation points in $\widetilde{\balpha}$. Straightforward computations give the lower bound on $q$ as stated in the proposition.
\end{proof}


Theorem~\ref{thm:alg} is obtained by investigating the algorithm implied by the proof of Proposition~\ref{prop:exist}. 

    \begin{algorithm}
\caption{Construction of $\balpha\in \Fq^{2k}$} \label{alg:construct}
\begin{algorithmic}[1]
\REQUIRE{$\alpha_1, \alpha_2, \alpha_3, \alpha_4 \in \mathbb{F}_q$}
\FOR{$i = 3$ to $k$} \label{ln:outer-loop}
    \STATE Initialize $\mathcal{B} \gets \emptyset$.
    \FORALL{pairs of increasing sequences $I^{\star}, J^{\star} \in [2i-2]^{2i-3}$ and $\widetilde{g}_{i-1} \in \mathbb{F}_q$} \label{ln:loop-for-lin}
        \STATE Solve the linear system in \eqref{eq:linear-system-for-alg} to obtain $(\widetilde{f}, \widetilde{g})$. \label{ln:root-finding} \label{ln:lin-sys}
        \STATE Solve each of the $5$ systems in~\eqref{eq:bad-pairs-computations}. \label{ln:root-find}
        \STATE Add to $\mathcal{B}$ every pair $(\beta_1, \beta_2)$ which is a solution to one of these systems.
    \ENDFOR
    \STATE Find $(\alpha_{2i-1}, \alpha_{2i}) \in \mathbb{F}_q^2 \setminus \mathcal{B}$ such that $\alpha_{2i-1}, \alpha_{2i} \notin \{\alpha_1, \ldots, \alpha_{2i-2}\}$. \label{ln:good-line}
    \STATE Output $(\alpha_1, \ldots, \alpha_{2k})$
\ENDFOR
\end{algorithmic}
\end{algorithm}    

\AlgHalfCode*
\begin{proof}
    We prove this theorem by describing explicitly the algorithm implied by the proof of Proposition~\ref{prop:exist}.
    First, let $q \geq 100 k^4$ be a prime power and define the $[4,2]_q$ RS code according to Remark~\ref{rem:rs-2-4}. Specifically, define $\alpha_1 = 0, \alpha_2 = 1$ and $\alpha_3, \alpha_4$ according to the constraints defined in Remark~\ref{rem:rs-2-4} to ensure that this $[4,2]_q$ RS code can correct a single insdel error. Now, run the algorithm given in Algorithm~\ref{alg:construct}.
    The correctness of this algorithm follows from Proposition~\ref{prop:exist}. 
    
    We now analyze the running time.
    The outer loop, in line~\ref{ln:outer-loop}, runs for $k-2$ iterations. 
    The inner loop in line~\ref{ln:loop-for-lin} runs for $O(i^2\cdot q) \leq O(k^2\cdot q)$. Indeed, there are at most $\binom{2i-2}{2}$ options for $I^{\star}, J^{\star}$ and $q$ options for $\widetilde{g}_{i-1}$. Note here that computing all the pairs $I^{\star}, J^{\star}\in [2i-2]^{2i-3}$ can be done in $O(i^3)$. Indeed, we need to run over all pairs of elements $(i,j)\in [2i-2]$ and output $I^{\star}, J^{\star} = [2i-2]\setminus\{i\}, [2i-2]\setminus\{j\}$.
    
    Inside the loop, in line~\ref{ln:lin-sys}, we need to solve a linear system which has at most one solution, and this takes $O(k^3)$ time. Then, in line~\ref{ln:root-find}, we have $5$ systems of equations, and we need to solve each one separately. One can verify that each of these systems can be solved by applying twice a root-finding algorithm for a degree $k-1$ polynomial. Thus, the entire inner loop takes $\textup{poly}(k, q)$.

    Finally, in line~\ref{ln:good-line}, we go over all $\binom{q}{2}$
    possible pairs and find a \textit{good} pair. Note that this step takes $O(q^2)$ and since $k=\Theta (q^4)$, the theorem follows.
\end{proof}


\section{Discussion and Future Work}
In this paper, we studied Reed-Solomon codes in the presence of insertion and deletion errors. Specifically, we investigated full-length Reed-Solomon codes and demonstrated that, when these codes have dimension 2, almost any ordering of the elements of $\mathbb{F}_q$ results in a code that can correct at least a single insertion or deletion error. Furthermore, we proved that for sufficiently large field size~$q$, nearly all full-length $2$-dimensional Reed-Solomon codes can correct up to $(1 - \delta)q$ insertion and deletion errors for any $0 < \delta < 1$.  Finally, using a probabilistic argument, we showed that if $q$ is large, there exists an ordering of $\mathbb{F}_q$ such that the $k$-dimensional Reed-Solomon code, with this ordering as the evaluation vector, can correct up to $q/(10k)$ insertion and deletion errors.  

In the second part of the paper, we investigated Reed-Solomon codes with a rate of $1/2$. By the half-Singleton bound, such codes can correct at most a single insertion or deletion error. We provided an existence result for codes with a rate of $1/2$ by induction on the dimension $k$, proving that \emph{optimal} codes---those capable of correcting the maximum number of insertion and deletion errors allowed by the half-Singleton bound---always exist when the underlying field size satisfies $q = O(k^4)$. The proof of this result also led to a deterministic algorithm for constructing such codes, which runs in polynomial time.  

Although we made progress toward a better understanding of Reed-Solomon codes in the context of insertion and deletion errors, several intriguing problems remain open. A natural question arising from our results is to better understand which orderings of the elements of the finite field $\mathbb{F}_q$ yield a full-length Reed-Solomon code that performs \emph{well} against insertion and deletion errors. In particular, even though we know that a random full-length 2-dimensional Reed-Solomon code will be able to correct $(1 - \delta)q$ insertion and deletion errors for any $0 < \delta < 1$ as long as $q$ is large enough, we do not know how to construct such a code explicitly. 

Moreover, the approach outlined in Section~\ref{sec:ex} does not appear to generalize in an obvious way to longer Reed-Solomon codes (with rates smaller than $1/2$). Current sufficient conditions on $q$ for the existence of such codes seem too restrictive,
and we anticipate to get a better understanding on which field size is actually required for the existence of effective Reed-Solomon codes against insertion and deletion errors.

\bigskip
\bigskip

\bibliographystyle{IEEEtran}
\bibliography{ourbib}

\end{document}